\documentclass[largeformat]{interact}

\usepackage{epstopdf}
\usepackage[caption=false]{subfig}

\usepackage[numbers,sort&compress]{natbib}
\bibpunct[, ]{[}{]}{,}{n}{,}{,}
\renewcommand\bibfont{\fontsize{10}{12}\selectfont}
\makeatletter
\def\NAT@def@citea{\def\@citea{\NAT@separator}}
\makeatother

\theoremstyle{plain}
\newtheorem{theorem}{Theorem}[section]

\newtheorem{corollary}[theorem]{Corollary}

\newtheorem{assumption}[theorem]{Assumption}

\theoremstyle{definition}

\newtheorem{problem}[theorem]{Problem}

\theoremstyle{remark}
\newtheorem{remark}{Remark}

\usepackage{graphicx}          

\usepackage{epsfig} 
\usepackage{amsfonts}
\usepackage{amsmath,amssymb}
\usepackage{comment}
\usepackage{color}
\usepackage{cases}
\usepackage{ascmac}
\usepackage{cuted}
\usepackage{enumerate}
\usepackage{here}

\newcommand{\calA}{{\mathcal A}}

\newcommand{\calF}{{\mathcal F}}

\newcommand{\calN}{{\mathcal N}}

\newcommand{\bbE}{{\mathbb E}}

\newcommand{\bbP}{{\mathbb P}}    
    
\newcommand{\bbR}{{\mathbb R}}

\newcommand{\bbU}{{\mathbb U}}    
    
\newcommand{\bbW}{{\mathbb W}}    
\newcommand{\bbX}{{\mathbb X}}    
    
\newcommand{\bbZ}{{\mathbb Z}}

\newcommand{\bmy}{{\mathbf y}}

\newcommand{\diag}{\mathop{\rm diag}\nolimits}

\newcommand{\rmd}{{\rm d}}

\newcommand{\sfw}{{\sf w}}

\newcommand{\bbra}[1]{\ensuremath{[\![#1]\!]} }  

\newcommand{\pdf}{\rho}
\newcommand{\kl}[2]{D_{\rm KL} \left({#1}\|{#2}\right)}
\newcommand{\ft}{N}

\newcommand{\magenta}[1]{{\color{black}{#1}}} 

\begin{document}


\title{Kullback-Leibler control for discrete-time nonlinear systems\\on continuous spaces}

\author{
\name{Kaito Ito\textsuperscript{a}\thanks{CONTACT Kaito Ito. Email: ito.kaito@bode.amp.i.kyoto-u.ac.jp} and Kenji Kashima\textsuperscript{a}}
\affil{\textsuperscript{a}Graduate School of Informatics, Kyoto University,
Kyoto, Japan}
}

\maketitle

\begin{abstract}
  Kullback-Leibler (KL) control enables efficient numerical methods for nonlinear optimal control problems. The crucial assumption of KL control is the full controllability of the transition distribution. However, this assumption is often violated when the dynamics evolves in a continuous space.
  Consequently, applying KL control to problems with continuous spaces requires some approximation, which leads to the lost of the optimality.
  To avoid such approximation, in this paper, we reformulate the KL control problem for continuous spaces so that it does not require unrealistic assumptions.
  The key difference between the original and reformulated KL control is that the former measures the control effort by KL divergence between controlled and {\em uncontrolled} transition distributions while the latter replaces the uncontrolled transition by a {\em noise-driven} transition.
  We show that the reformulated KL control admits efficient numerical algorithms like the original one without unreasonable assumptions.
  \magenta{Specifically}, the associated value function can be computed by using a Monte Carlo method based on its path integral representation.
\end{abstract}


\section{Introduction}
Optimal control theory is a powerful mathematical tool for achieving control objectives while considering, for example, energy efficiency and sparsity of control~\cite{Lewis2012,Ito2021}.
Optimal control problems arise in a variety of physical, biological, and economic systems, to name a few. 
Recently, optimal control has also become increasingly important in machine learning~\cite{Liu2019,Recht2019}.
It is well-known that finding an optimal feedback control law boils down to solving the (Hamilton-Jacobi) Bellman equation~\cite{Hernandez1996,Yong1999}, which suffers from the curse of dimensionality and is difficult to solve in general.

On the other hand, a special class of stochastic optimal control problems was introduced in~\cite{Kappen2005linear,Todorov2006}, in which the associated Bellman equation can be converted into a linear equation resulting in efficient numerical methods.
For continuous state/input spaces and continuous time, the work~\cite{Kappen2005linear} considers a control-affine diffusion with a quadratic control cost and assumes the noise and control act in the same subspace.
Then, the optimal control admits a path integral representation, which can be approximated by forward sampling of an uncontrolled diffusion process. This stochastic control framework is called a path integral control and has many applications, e.g., reinforcement learning~\cite{Theodorou2010,Theodorou2010b}, model predictive control~\cite{Williams2017}, multi-agent systems~\cite{Van2008}, controllability quantification~\cite{Kashima2016}.

For discrete-time cases, the work~\cite{Todorov2006} deals with general dynamics and makes the key assumptions as follows: (A1) the controller can change the distribution of the next state given the current state as desired; (A2) the control cost is quantified by the Kullback-Leibler (KL) divergence between the controlled and uncontrolled state distributions.
This formulation is referred to as linearly-solvable Markov decision processes (MDPs) or KL control.
The KL control framework shares the nice properties with the path integral control including a path integral representation of the KL optimal control~\cite{Todorov2009}, compositionality of optimal control laws~\cite{TOdorov2009compositionality}, and duality with Bayesian inference~\cite{Todorov2008}.
For the connections of the path integral control and KL control, see~\cite{Theodorou2012}.
Moreover, the special structure of KL control enables the convex formulation of inverse reinforcement learning~\cite{Dvijotham2010}.

However, it should be emphasized that the assumption (A1) of KL control is too restrictive in practice, especially for continuous state spaces. Indeed, as mentioned in \cite{Rawlik2012}, even for discrete-time linear systems driven by Gaussian noise, (A1) is violated because the variance of the one step transition distribution given the current state is uncontrollable under the causality of controllers.
Therefore, applying KL control to practical problems with continuous spaces requires some approximation, which leads to the lost of the optimality.
For instance, if the system of interest is derived from the Euler-Maruyama discretization of a control-affine diffusion, using a smaller time step results in the smaller approximation error~\cite{Todorov2009eigen,Zhong2011}.
However, to the best of our knowledge, there is no discussion of approximation in other cases, e.g., the system originally evolves in discrete time.

{\it Contribution:}
In this paper, we reformulate KL control for continuous state spaces so that its assumption is more realistic than the conventional formulation of KL control.
This enables us to apply KL control to discrete-time and continuous space problems without any approximation of dynamics.
As a byproduct, we reconsider what the assumption (A1) implies for practical problems.
Specifically, we clarify that (A1) essentially requires the controller to know the value of noise to be injected to the system together with control inputs.
Moreover, we show that our KL control formulation enjoys the nice properties which the original one has as mentioned above.

{\it Organization:}
The remainder of this paper is organized as follows. In Section~\ref{sec:review}, we briefly review KL control.
In Section~\ref{sec:problem}, we reformulate KL control for continuous spaces. Section~\ref{sec:analysis} is devoted to the general analysis of the reformulated KL control. In Section~\ref{sec:LQG}, we focus on linear systems with a quadratic state cost.
In Section~\ref{sec:example}, numerical examples are presented.
Some concluding remarks are given in Section~\ref{sec:conclusion}.

{\it Notation:}
Let $ \bbR $ denote the set of real numbers and $ \bbZ_{>0} $ (resp. $ \bbZ_{\ge 0} $) denote the set of positive (resp. nonnegative) integers.
The set of integers $ \{0,1,\ldots,N\} $ is denoted by $ \bbra{N} $.
The identity matrix is denoted by $ I $, and its dimension depends on the context.
For a matrix $ A\in \bbR^{n\times n} $, we write $ A \succ 0$ if $ A $ is symmetric and positive definite. 
The determinant of a square matrix $ A $ is denoted by $ {\rm det}(A) $.
The block diagonal matrix with diagonal entries $ \{A_i\}_{i=1}^N, A_i \in \bbR^{m\times n} $ is denoted by $ \diag (A_1,\ldots,A_N) $.
Let $(\Omega, \calF, \bbP)$ be a complete probability space where $ \calF $ is the $ \sigma $-field on $ \Omega $, and $ \bbP : \calF \rightarrow [0,1] $ is a probability measure. 
The space $(\Omega, \calF, \bbP)$ is equipped with a natural filtration $\{\calF_k\}_{k\ge 0}$. The expectation is denoted by $\bbE$.
The probability density function of a continuous random variable $ x $ with respect to the Lebesgue measure is denoted by $ \pdf_x $.
The conditional density of $ x $ given $ y = \bmy $ is denoted by $ \pdf_{x|y} (\cdot | \bmy) $.
Denote by $ \kl{\pdf_x}{\pdf_y} $ the KL divergence between probability densities $ \pdf_x $ and $ \pdf_y $.
The Dirac delta function is denoted by $ \delta(\cdot) $.
For an $ \bbR^n $-valued random vector $ w $, $ w\sim \calN(\mu,\Sigma) $ means that $ w $ has a multivariate Gaussian distribution with mean $ \mu \in \bbR^n $ and covariance matrix $ \Sigma $. When $ \Sigma \succ 0 $, the density function of $ w\sim \calN (\mu,\Sigma) $ is denoted by $ \calN (\cdot | \mu, \Sigma) $.

\section{Brief introduction of KL control}\label{sec:review}
Here, we briefly review KL control~\cite{Todorov2006}.
Let $ \bbX \subseteq \bbR^n $ be a state space and $ \bbU \subseteq \bbR^m $ an input space.
Consider an MDP with a transition density function $ \pdf_{x_{k+1}|x_k, u_k} $ where $ \{x_k\} $ is an $ \bbX $-valued state process, and $ \{u_k\} $ is a $ \bbU $-valued control process.
In this section, we implicitly assume the existence of probability density functions. Nevertheless, we can apply the same argument for discrete random variables by replacing densities by probabilities.
Let $ \pdf_{x_0} $ be the density of the initial state $ x_0 $.
Denote by $ \pdf_{k+1}^{\pi_k} (\cdot | x) $ the conditional density of $ x_{k+1} $ given $ x_k = x $ induced by a stochastic policy (control law) $ \pi_k (\cdot | x) := \pdf_{u_k|x_k} (\cdot | x ) $. Let $ \pdf_{k+1}^0 (\cdot|x) := \pdf_{x_{k+1}|x_k,u_k}(\cdot | x, 0) $ be the transition density for the uncontrolled dynamics.
Then, the KL control problem is formulated as follows.
\begin{problem}\label{prob:conventional_KL}
  Given a finite horizon $ \ft \in \bbZ_{>0} $, find policies $ \{\pi_k\}_{k=0}^{\ft-1} $ that solve
  \begin{equation}\label{eq:conventional_KLcontrol}
    \begin{aligned}
      &\underset{\{\pi_k\}_{k=0}^{\ft-1}}{\rm minimize} ~~ \bbE \biggl[ \ell_\ft (x_\ft) + \sum_{k=0}^{\ft-1} \Bigl\{ \ell_k (x_k) + \kl{\pdf_{k+1}^{\pi_k} (\cdot | x_k)}{\pdf_{k+1}^0 ( \cdot | x_k)} \Bigr\} \biggr],
      \end{aligned}
  \end{equation}
  where $ \ell_k : \bbR^n \rightarrow \bbR $ is the running cost ($ k = 0,\ldots,\ft-1 $) and terminal cost ($ k = \ft $) for the state, respectively.
  \hfill$ \diamondsuit $
\end{problem}
Here, we assume that under $ u_k = 0, \forall k $, $ \bbE[\ell_k (x_k)] $ takes finite values for all $ k\in \bbra{N} $.
KL divergence measures the difference between two probability distributions.
Hence, Problem~\ref{prob:conventional_KL} penalizes the deviation of the transition density $ \pdf_{k+1}^{\pi_k} (\cdot | x_k) $ from the uncontrolled transition density $ \pdf_{k+1}^0 (\cdot | x_k) $.

Now, we introduce the most important assumption of KL control.
\begin{assumption}\label{ass:conventional_KL}
  For any $ x \in \bbX, \ k\in \bbra{\ft-1} $ and any density $ \check{\pdf} $ on $ \bbX $, there exists a policy $ \pi_k $ such that $ \check{\pdf} (x') = \pdf_{k+1}^{\pi_k} (x' | x)$ for all $ x'\in \bbX $.
  \hfill $ \diamondsuit $
\end{assumption}
The above assumption says that the controller can change the transition density as desired.
Under this assumption, the Bellman equation for \eqref{eq:conventional_KLcontrol} becomes linear by an exponential transformation:
\begin{align}
  &z(k,x) = \exp (-\ell_k (x)) \calA_{\pdf_{k+1}^0} [z] (k,x), \  k \in \bbra{N-1}, \ x\in \bbX, \label{eq:conventional_Z_bellman} \\
  &z(\ft,x) = \exp (-\ell_\ft (x)), \ x\in \bbX , \label{eq:conventional_Z_bellman_terminal}
\end{align}
where $ \calA_{\pdf_{k+1}^0} [z] (k,x) := \int_\bbX z(k+1,x') \pdf_{k+1}^0 (x'|x) \rmd x' $.
The solution of \eqref{eq:conventional_Z_bellman},~\eqref{eq:conventional_Z_bellman_terminal} is given by the so-called desirability function $ z(k,x) := \exp(-v(k,x)) $, and the value function $ v $ associated with \eqref{eq:conventional_KLcontrol} is defined by
\begin{align*}
  &v(k,x) := \inf_{\{\pi_s\}_{s=k}^{\ft-1} } \bbE \biggl[ \ell_\ft (x_\ft) + \sum_{s=k}^{\ft-1} \Bigl\{ \ell_s (x_s) + \kl{\pdf_{s+1}^{\pi_s} (\cdot | x_s)}{\pdf_{s+1}^0 ( \cdot | x_s)} \Bigr\} \biggl| \ x_k = x \biggr], \\
  &\hspace{12cm} k \in \bbra{\ft} ,\ x \in \bbX .
\end{align*}
In particular, a policy $ \pi_k^* $ satisfying
\begin{equation}\label{eq:conventional_optimality}
  \pdf_{k+1}^{\pi_k^*} (x' | x) = \frac{\pdf_{k+1}^0(x' | x) z(k+1,x')}{\calA_{\pdf_{k+1}^0} [z](k,x)}, \ \forall x',x \in \bbX
\end{equation}
is an optimal policy, and its existence is ensured by Assumption~\ref{ass:conventional_KL}.
It is remarkable that an optimal transition density can be written analytically given the desirability function unlike the conventional MDPs~\cite{Hernandez1996}.
However, as mentioned in the Introduction, Assumption~\ref{ass:conventional_KL} is typically violated for continuous state spaces, and there is no policy satisfying~\eqref{eq:conventional_optimality}.

\section{Reformulation of KL control for continuous spaces}\label{sec:problem}
In this section, we reformulate KL control to make its assumption more realistic for continuous spaces. The same notation as in Section~\ref{sec:review} is employed.
In this paper, we consider general nonlinear systems of the form:
\begin{align}
  &x_{k+1} = f(x_k, u_k), \ k \in \bbZ_{\ge 0}, \label{eq:system}\\
  &x_0 \sim \pdf_{x_0}, \label{eq:system_initial}
\end{align}
where $ \{x_k\} $ is an $ \bbX $-valued state process, $ \{u_k\} $ is a $ \bbU $-valued control process, and $ f : \bbX\times \bbU \rightarrow \bbX $.
Next, we introduce the associated noise-driven dynamics:
\begin{align}
  &\bar{x}_{k+1} = f(\bar{x}_k, w_k), \ k\in \bbZ_{\ge 0}, \label{eq:noise_driven} \\
  &\bar{x}_0 \sim \pdf_{x_0},
\end{align}
where $\{w_k\} $ is a sequence of independent (not necessarily identically distributed) random variables, and $ w_k $ has the density function $ \pdf_{w_k} $ with the support $ \bbW $.
Denote the conditional density of $\bar{x}_{k+1}$ given $ \bar{x}_k = x $ by $ \bar{\pdf}_{k+1} (\cdot |  x) $.
Now, we are ready to state our problem.
\begin{problem}\label{prob:KLcontrol}
  Given a finite horizon $ \ft \in \bbZ_{>0} $, find policies $ \{\pi_k\}_{k=0}^{\ft -1} $ that solve
  \begin{equation}\label{eq:cost}
    \begin{aligned}
      &\underset{\{\pi_k\}_{k=0}^{\ft-1}}{\rm minimize} ~~ \bbE \biggl[ \ell_\ft (x_\ft) + \sum_{k=0}^{\ft-1} \Bigl\{ \ell_k (x_k) + \kl{\pdf_{k+1}^{\pi_k} (\cdot | x_k)}{\bar{\pdf}_{k+1} ( \cdot | x_k)} \Bigr\} \biggr] \\
      &\text{subject to} ~~ \eqref{eq:system}, \eqref{eq:system_initial} .
      \end{aligned}
  \end{equation}
  \hfill $ \diamondsuit $
\end{problem}

We emphasize that Problem~\ref{prob:KLcontrol} employs {\em noise-driven} dynamics ($ u_k = w_k $) as a reference transition density $ \bar{\pdf}_{k+1} (\cdot | x_k) $ while Problem~\ref{prob:conventional_KL} employs {\em uncontrolled} dynamics ($ u_k = 0 $).
Note that for a deterministic policy $ u_k = K (x_k) $, i.e., $ \pi_k (u | x) = \delta (u - K(x)) $, $\kl{\pdf_{k+1}^{\pi_k} (\cdot | x_k) }{\bar{\pdf}_{k+1} (\cdot | x_k)} $ is infinite because $ \pdf_{k+1}^{\pi_k} (\cdot | x_k)$ is not absolutely continuous with respect to $ \bar{\pdf}_{k+1} (\cdot | x_k) $.
Therefore, an optimal policy for Problem~\ref{prob:KLcontrol} must be stochastic.
This is in contrast to the conventional optimal control problems without the KL divergence cost whose optimal policy is deterministic~\cite{Hernandez1996}.

For $ x\in \bbX $, let $ f_x(u) := f(x,u) $ and $ \bbX_x := \{f_x (u) : u\in \bbU\} $.
Then, we assume the following conditions.
\begin{assumption}\label{ass:invertible}
  \hspace{\fill}
  \begin{itemize}
  \item[(i)] $ \bbW \subseteq \bbU $;
  \item[(ii)] $ m = n $;
  \item[(iii)] For all $ x\in \bbX $, $ f_x : \bbU \rightarrow \bbX_x $ is bijective and continuously differentiable. In addition, for all $ x\in \bbX $, the Jacobian matrix $ {\rm J}_{f_x^{-1}}  $ of the inverse function $ f_x^{-1} $ satisfies
  \[ 
     \left| {\rm det}\left({\rm J}_{f_x^{-1}} (x') \right) \right|  \neq 0, \ \forall x' \in \bbX_x ;
  \]
  \item[(iv)] For all $ k \in \bbra{N} $, $ \bbE[\ell_k (\bar{x}_k)] $ takes finite values.
  \hfill $ \diamondsuit $
  \end{itemize}
\end{assumption}
Assumptions~\ref{ass:invertible}-(ii),(iii) ensure the existence of the density $ \rho_{k+1}^{\pi_k} (\cdot | x) $~\cite[Chapter~6, Theorem~5]{Roussas2015}.
In addition, Assumption~\ref{ass:invertible}-(i) means that there exists a feasible control process that replicates a given noise process.
Consequently, the transition density $ \pdf_{k+1}^{\pi_k} (\cdot | x) $ can be shaped to a desired form with the support $ \{ f_x(w) : w\in \bbW \} $ by an appropriate policy; see the proof of Theorem~\ref{thm:opt_ctrl_finite} in the next section.
Therefore, Assumption~\ref{ass:invertible}-(i) corresponds to Assumption~\ref{ass:conventional_KL}.
Lastly, Assumption~\ref{ass:invertible}-(iv) is a technical assumption that ensures there exists a policy that makes \eqref{eq:cost} finite. For instance, if $ \ell_k $ is bounded for all $ k\in \bbra{N} $, Assumption~\ref{ass:invertible}-(iv) is satisfied.

\begin{remark}
Consider the control-affine case $ f(x,u) = f_0 (x) + g(x) u $ where $ f_0 : \bbR^n \rightarrow \bbR^n, \ g: \bbR^n \rightarrow \bbR^{n\times m} $. Then, Assumptions~\ref{ass:invertible}-(ii),(iii) imply that, for all $ x\in \bbR^n $, $ g(x) $ is square and invertible.
Note that when $ n < m $ and $ g(x) $ has full row rank for all $ x\in \bbR^n $, we can introduce an auxiliary system
\begin{equation}\label{eq:auxiliary}
        \tilde{x}_{k+1} = \tilde{f}_0 (\tilde{x}_k) + \tilde{g} (\tilde{x}_k) u_k
\end{equation}
where $ \tilde{x}_k \in \bbR^{m-n},\ \tilde{f}_0 : \bbR^{m-n} \rightarrow \bbR^{m-n},\ \tilde{g} : \bbR^{m-n} \rightarrow \bbR^{(m-n)\times m} $, such that the combined system
\begin{equation}
   \begin{bmatrix}
           x_{k+1} \\ \tilde{x}_{k+1}
   \end{bmatrix} = 
   \begin{bmatrix}
           f_0 (x_k) \\ \tilde{f}_0 (\tilde{x}_k)
   \end{bmatrix} + 
   \begin{bmatrix}
           g(x_k) \\ \tilde{g} (\tilde{x}_k)
   \end{bmatrix} u_k
\end{equation}
satisfies Assumptions~\ref{ass:invertible}-(ii),(iii). That is, $ [g(x)^\top \ \tilde{g}  (\tilde{x})^\top ]^\top $ is invertible for all $ [x^\top \ \tilde{x}^\top ]^\top \in \bbR^m $.
When the state cost function $ \ell_k $ does not depend on $ \tilde{x}_k $, the introduction of the auxiliary system~\eqref{eq:auxiliary} is explicitly relevant only for the KL divergence cost of \eqref{eq:cost}.
\hfill $ \diamondsuit $
\end{remark}

\section{General analysis of KL control for continuous spaces}\label{sec:analysis}
In this section, we characterize the value function and the optimal control of Problem~\ref{prob:KLcontrol} and then reconsider the implication of Assumption~\ref{ass:conventional_KL} for Problem~\ref{prob:conventional_KL}.

\subsection{\magenta{Characterization of value function and optimal control}}
Define the value function associated with \eqref{eq:cost} as follows:
\begin{align*}
  &V(k,x) := \inf_{\{\pi_s\}_{s=k}^{\ft-1} } \bbE \biggl[ \ell_\ft (x_\ft) + \sum_{s=k}^{\ft-1} \Bigl\{  \ell_s (x_s) + \kl{\pdf_{s+1}^{\pi_s} (\cdot | x_s)}{\bar{\pdf}_{s+1} ( \cdot | x_s)} \Bigr\} \biggl| \ x_k = x \biggr], \\
  &\hspace{12cm} k \in \bbra{\ft} ,\ x \in \bbX .
\end{align*}
Then the optimal value for Problem~\ref{prob:KLcontrol} is given by $ \bbE[V(0,x_0)] $.
Also, define the desirability function
\begin{equation}
  Z(k,x) := \exp (-V(k,x)) .
\end{equation}
Similarly to the conventional optimal control, the desirability function or, equivalently, the value function plays a crucial role in our problem.

\begin{theorem}\label{thm:opt_ctrl_finite}
  Suppose that Assumption~\ref{ass:invertible} holds.
  Then, the unique optimal policy $ \{\pi_k^* \} $ for Problem~\ref{prob:KLcontrol} is given by
  \begin{equation}\label{eq:opt_policy}
    \pi_k^* (u | x) := \frac{\pdf_{w_k} (u) Z(k+1, f(x,u))}{\calA_{\bar{\pdf}_{k+1}} [Z] (k, x)}, \ u\in \bbU, \ x\in \bbX .
  \end{equation}
  In addition, the desirability function $ Z $ satisfies
  \begin{align}
    &Z(k,x) = \exp (-\ell_k (x)) \calA_{\bar{\pdf}_{k+1}} [Z] (k, x), \ k \in \bbra{N-1}, \ x\in \bbX, \label{eq:Z_eq} \\
    &Z(\ft,x) = \exp (-\ell_\ft (x)), \ x\in \bbX \label{eq:Z_terminal} .
  \end{align}
\end{theorem}
\begin{proof}
By the dynamic programming principle, the value function $ V $ satisfies the Bellman equation
\begin{align}
  &V(k,x) = \ell_k (x) + \inf_{\pi_k} \bigl\{ \kl{\pdf_{k+1}^{\pi_k} (\cdot | x)}{\bar{\pdf}_{k+1} (\cdot | x)} + \calA_{\pdf_{k+1}^{\pi_k}} [V] (k,x) \bigr\}, \ k\in \bbra{\ft-1},\ x\in \bbX , \label{eq:Bellman_eq} \\
  &V(\ft,x) = \ell_\ft (x), \ x\in \bbX .
\end{align}
In addition, if a policy $ \pi_k $ achieves the minimum of the right-hand side of \eqref{eq:Bellman_eq}, this is an optimal policy. Note that
\begin{align}
  &\kl{\pdf_{k+1}^{\pi_k} (\cdot | x)}{\bar{\pdf}_{k+1} (\cdot | x)} + \calA_{\pdf_{k+1}^{\pi_k}} [V] (k,x) \nonumber \\
  &\quad = \int_{\bbX} \pdf_{k+1}^{\pi_k} (x' | x) \log \frac{\pdf_{k+1}^{\pi_k} (x'|x)}{\bar{\pdf}_{k+1} (x'|x) Z(k+1, x')} \rmd x' \nonumber \\
  &\quad = \kl{\pdf_{k+1}^{\pi_k}(\cdot | x)}{\pdf_{k+1}^* (\cdot | x)} - \log \calA_{\bar{\pdf}_{k+1}} [Z] (k, x) , \label{eq:inf}
\end{align}
where we defined
\begin{align}
  &\pdf_{k+1}^* (x' | x) := \frac{\bar{\pdf}_{k+1} (x'|x) Z(k+1,x')}{\calA_{\bar{\rho}_{k+1}} [Z] (k,x)}, \ x',x \in \bbX .
\end{align}
The second term in the right-hand side of \eqref{eq:inf} does not depend on $ \pi_k $.
Therefore, if a policy $ \pi_k $ satisfies
\begin{equation}\label{eq:opt_policy_condition}
  \pdf_{k+1}^{\pi_k} (x' | x) = \pdf_{k+1}^* (x' | x), \ \forall x,x' \in \bbX, 
\end{equation}
this is an optimal policy at time $ k $.
For any $ x\in \bbX $, by Assumption~\ref{ass:invertible} and the change of variables $ x' = f_x(u) $ for $ \pi_k (u | x) $~\cite[Chapter~6, Theorem~5]{Roussas2015}, we obtain
\begin{equation*}
  \pdf_{k+1}^{\pi_k} (x' | x) = \pi_k \left(f_x^{-1} (x') | x \right) \left| {\rm det}\left({\rm J}_{f_x^{-1}} (x') \right) \right|.
\end{equation*}
Similarly, we have
\begin{equation*}
  \bar{\pdf}_{k+1} (x' | x) = \pdf_{w_k} \left(f_x^{-1} (x') \right) \left| {\rm det}\left({\rm J}_{f_x^{-1}} (x') \right) \right|.
\end{equation*}
Therefore, $ \pi_k^* $ defined in \eqref{eq:opt_policy} is a unique policy satisfying \eqref{eq:opt_policy_condition}.
As a result, the Bellman equation \eqref{eq:Bellman_eq} can be simplified as
\begin{align}
        V(k,x) = \ell_k (x) - \log \calA_{\bar{\pdf}_{k+1}} [Z] (k,x) ,
\end{align}
which completes the proof.
\end{proof}

From Theorem~\ref{thm:opt_ctrl_finite}, similarly to the conventional optimal control, Problem~\ref{prob:KLcontrol} boils down to calculating the desirability function $ Z $.
A notable difference between them is that thanks to the linearity of \eqref{eq:Z_eq}, the desirability function for KL control admits the path integral representation.
\begin{corollary}\label{cor:path_integral}
  Suppose that Assumption~\ref{ass:invertible} holds.
  Then, the desirability function $ Z $ satisfies
  \begin{equation}
      Z (k, x) = \bbE\left[ \exp \left( - \sum_{s=k}^{\ft} \ell_s (\bar{x}_s)		\right)	 \middle| \ \bar{x}_k = x	\right], \label{eq:path_integral}
  \end{equation}
  where $ \{\bar{x}_s \} $ is a solution of \eqref{eq:noise_driven}.
\end{corollary}
\begin{proof}
  By using \eqref{eq:Z_eq},~\eqref{eq:Z_terminal}, and induction on $ k $, we immediately obtain the desired result.
\end{proof}

The path integral representation~\eqref{eq:path_integral} motivates us to compute the desirability function by sampling approximations.
In particular, if one can simulate sample paths of $ \{\bar{x}_k \} $, the sampling approximations of \eqref{eq:path_integral} do not require the knowledge of $ f $. Hence, \eqref{eq:path_integral} enables model-free approaches for obtaining the optimal policy.

Next, we consider the discrete input space $ \bbU = \{u^{(i)}\}_{i=1}^{r}, \ u^{(i)} \in \bbR^m, \ r\in \bbZ_{> 0} \cup \{\infty\} $.
In this case, density functions must be replaced by probabilities such as a policy $ \Pi_k (u^{(i)} | x) := \bbP(u_k  = u^{(i)} | x_k = x) $.
Then, we obtain the following.
\begin{corollary}\label{cor:discrete}
  Suppose that Assumptions~\ref{ass:invertible}-(i),(iv) hold.
  Then, for Problem~\ref{prob:KLcontrol} with $ \bbU = \{u^{(i)}\}_{i=1}^{r} $, there exists a policy $ \{\Pi_k^*\} $ such that for all $ x\in \bbX,\ x'\in \bar{\bbX}_x := \{f(x,w) : w\in \bbW\} $, it holds
  \begin{align}
    &\sum_{i : f(x,u^{(i)}) = x'} \Pi_k^* (u^{(i)} | x)  = \frac{\bbP \left( f(x, w_k) = x' \right) Z(k+1,x') }{\sum_{x'' \in \bar{\bbX}_x} \bbP \left(f(x,w_k) = x'' \right) Z(k+1,x'' )} . \label{eq:opt_policy_discrete}
  \end{align}
Here, the desirability function $ Z $ satisfies \eqref{eq:Z_eq},~\eqref{eq:Z_terminal} where $ \calA_{\bar{\pdf}_{k+1}}[Z](k,x) $ is replaced by $ \sum_{x'\in \bar{\bbX}_x} Z(k+1,x') \bbP(f (x,w_k) = x') $, and admits the representation~\eqref{eq:path_integral}. In addition, $ \{\Pi_k^*\} $ is an optimal policy for Problem~\ref{prob:KLcontrol}.
Furthermore, if for all $ x \in \bbX$, $ f_x : \bbU \rightarrow \bbX_x $ is bijective, a policy satisfying \eqref{eq:opt_policy_discrete} is a unique optimal policy.
\end{corollary}
\begin{proof}
  Note that
  \begin{align}
    &\bbP(x_{k+1} = x' | x_k = x) =\sum_{i : f(x,u^{(i)}) = x'} \Pi_k^* (u^{(i)} | x), \\
    &\bbP(\bar{x}_{k+1} = x' | \bar{x}_k = x) = \bbP \left(f(x,w_k) = x' \right).
  \end{align}
  Then, by the same argument as in the proof of Theorem~\ref{thm:opt_ctrl_finite}, we obtain the existence and optimality of $ \{\Pi_k^*\} $ satisfying~\eqref{eq:opt_policy_discrete}. Especially when for all $ x\in \bbX $, $ f_x $ is bijective, $ \{u\in \bbU : f(x,u) = x'\} $ is a singleton for all $ x\in \bbX,\ x' \in \bar{\bbX}_x $, which leads to the uniqueness of the optimal policy.
\end{proof}
The above result clarifies that in Assumption~\ref{ass:invertible}, the condition~(i) $ \bbW\subseteq \bbU $ plays a crucial role in making optimal control problems linearly solvable while the bijectivity of $ f_x $ ensures the uniqueness of the optimal policy. Note that Corollary~\ref{cor:discrete} does not assume $ m = n $.

\subsection{Reconsideration of controllability assumption of transition density}
Now, let us go back to the original formulation of KL control (Problem~\ref{prob:conventional_KL}) and reconsider the implication of Assumption~\ref{ass:conventional_KL}.
In the rest of this section, the control-affine system is considered:
\begin{equation}
  x_{k+1} = f_0(x_k) + g(x_k) (u_k + w_k), \ w_k \sim \pdf_{w_k}, 
\end{equation}
where $\{w_k\} $ is a sequence of independent random variables.
For the above system, causal controllers $ \pi_k (u_k | x_k) $ cannot satisfy \eqref{eq:conventional_optimality}, and therefore the associated Bellman equation cannot be linearized.
To gain deeper insight into Assumption~\ref{ass:conventional_KL} that ensures the existence of a policy satisfying \eqref{eq:conventional_optimality}, we shall introduce an atypical assumption.
\begin{assumption}\label{ass:noncausality}
  The control input $ u_k $ is allowed to depend on $ w_k $.
  \hfill $ \diamondsuit $
\end{assumption}
This assumption means that the causality of controllers can be violated.
Now the decision variables for Problem~\ref{prob:conventional_KL} are replaced by $ \pi_{\sfw,k} (\cdot | x, w) := \pdf_{u_k|x_k,w_k} (\cdot | x, w), \ k\in \bbra{\ft-1} $.
Then, we have the following result.
\begin{theorem}
  Suppose that Assumptions~\ref{ass:invertible}-(ii),(iii),~\ref{ass:noncausality} hold for $ f_x(u) = f_0 (x) + g(x) u $.
  Then, the unique optimal policy for Problem~\ref{prob:conventional_KL} is given by
  \begin{align}
    \pi_{\sfw,k}^* (u | x, w) := \frac{\pdf_{w_k} (u + w) z\left(k+1, f_0 (x) + g(x)u\right)}{\calA_{\pdf_{k+1}^0} [z](k,x)}, \ u\in \bbU, \ x\in \bbX.
  \end{align}
  In addition, the desirability function $ z $ satisfies \eqref{eq:conventional_Z_bellman},~\eqref{eq:conventional_Z_bellman_terminal}.
\end{theorem}
\begin{proof}
  Note that
  \begin{equation}
    \rho_{k+1}^0 (x' | x) = \frac{1}{|{\rm det} \left( g(x)\right)|} \rho_{w_k} \left( \left(g(x)\right)^{-1} (x' - f(x))  \right)
  \end{equation}
  and under a policy $ \pi_{{\sf w},k} $,
  \begin{align}
     \pdf_{x_{k+1}|x_k,w_k} (x' | x,w) = \frac{1}{|{\rm det} \left( g(x)\right)|} \pi_{{\sf w},k} \left( \left(g(x) \right)^{-1} (x' - f(x)) - w \bigl| \ x,w \right).
  \end{align}
  Also, we have
  \begin{align}
    \pdf_{k+1}^{\pi_{{\sf w},k}} (x' | x) &= \int_{\bbW} \pdf_{x_{k+1},w_k|x_k} (x', w | x) \rmd w \nonumber\\
    &= \int_{\bbW} \pdf_{x_{k+1}|x_k,w_k} (x' | x,w) \pdf_{w_k|x_k} (w | x) \rmd w \nonumber\\
    &= \int_{\bbW} \pdf_{x_{k+1}|x_k,w_k} (x' | x,w) \pdf_{w_k} (w) \rmd w .
  \end{align}
  Then, it is straightforward to check that \eqref{eq:conventional_optimality} is satisfied for $ \pi_{\sfw,k} = \pi_{\sfw,k}^* $.
\end{proof}
This theorem shows that Assumption~\ref{ass:noncausality} for the noncausality of policies plays the same role as Assumption~\ref{ass:conventional_KL}.
Of course, the noncausality is unrealistic for practical applications.
This clarifies that the reformulated KL control is much more realistic for systems on continuous spaces than the original formulation of KL control.

\section{Linear quadratic Gaussian setting}\label{sec:LQG}
In this section, we focus on a linear system ($ f(x,u) = Ax+Bu $) with $ \bbU = \bbR^m $, a quadratic cost
\begin{equation}\label{eq:quadratic_cost}
 \ell_k (x) = \frac{1}{2} x^\top Q_k x, \ Q_k \succ 0 ,\ k = 0,\ldots,\ft ,
\end{equation}
and Gaussian noise $ w_k \sim \calN (0, \Sigma_k), \ \Sigma_k \succ 0 $.
Assume that $ m = n $ and $ B $ is invertible. Then Assumption~\ref{ass:invertible} is satisfied.
Now, we calculate the optimal policy for Problem~\ref{prob:KLcontrol} analytically.
First, for $ k = \ft -  1 $, we have
\begin{align}
  &\pi_{\ft-1}^* (u|x) \propto \calN(u | 0, \Sigma_{\ft-1}) Z (\ft, Ax + Bu)  \nonumber\\
  &\propto \exp \left( - \frac{1}{2} \left( u^\top \Sigma_{\ft-1}^{-1} u + (Ax + Bu)^\top Q_\ft (Ax + Bu) \right)   \right) \nonumber \\
  &\propto \exp \biggl( - \frac{1}{2} \left[ u + (\Sigma_{\ft-1}^{-1} + B^\top Q_\ft B)^{-1} B^\top Q_\ft Ax  \right]^\top \nonumber \\
  &\qquad\quad \times (\Sigma_{\ft-1}^{-1} + B^\top Q_\ft B)  \nonumber \\
  &\qquad\quad \times  \left[ u + (\Sigma_{\ft-1}^{-1} + B^\top Q_\ft B)^{-1} B^\top Q_\ft Ax  \right] \biggr),
\end{align}
which means that
\begin{align}
  \pi_{\ft-1}^* (u | x) &= \calN \bigl(u \bigl| -(\Sigma_{\ft-1}^{-1} + B^\top Q_\ft B)^{-1} B^\top Q_\ft Ax, (\Sigma_{\ft-1}^{-1} + B^\top Q_\ft B)^{-1} \bigr) .
\end{align}
On the other hand,
\begin{align}
  &Z(\ft - 1,x) = \exp \left( -\frac{1}{2} x^\top Q_{\ft -1} x  \right)  \int_{\bbR^n} \calN (x' | Ax, B\Sigma_{\ft-1} B^\top) Z(\ft, x') \rmd x' \nonumber \\
  &= [{\rm det} (I + Q_\ft B \Sigma_{\ft-1} B^\top )]^{-1/2} \nonumber \\
  & \times \exp\Biggl( -\frac{1}{2} x^\top \bigl( Q_{\ft-1} +  A^\top (I - (I + Q_\ft B\Sigma_{\ft-1} B^\top)^{-1}  )  (B\Sigma_{\ft-1} B^\top)^{-1} A   \bigr) x \Biggr) , \label{eq:Z_ft-1}
\end{align}
where we used the formula
\begin{align*}
  &\bbE\left[ \exp \left( -\frac{1}{2} x^\top Q x \right)  \right] \\
  &= [{\rm det} (I + Q\Sigma)]^{-1/2} \exp \left( -\frac{1}{2} \mu^\top (I - (I + Q\Sigma)^{-1}) \Sigma^{-1} \mu \right)
\end{align*}
for $ Q \succ 0 $ and $ x \sim \calN (\mu, \Sigma),  \Sigma \succ 0 $.
Note that
\begin{align*}
  &(I - (I + Q_\ft B\Sigma_{\ft-1} B^\top)^{-1}) (B\Sigma_{\ft-1} B^\top)^{-1} \\
  &= (I - \Sigma_B^{-1} (I + Q_\ft^{-1}\Sigma_B^{-1})^{-1} Q_\ft^{-1} ) \Sigma_B^{-1} \\
  &= (Q_\ft^{-1} + B\Sigma_{\ft-1} B^\top)^{-1} \\
  &= Q_\ft - Q_\ft B (\Sigma_{\ft-1}^{-1} + B^\top Q_\ft B)^{-1} B^\top Q_\ft ,
\end{align*}
where $ \Sigma_B := B\Sigma_{\ft-1} B^\top $.
Substituting this into \eqref{eq:Z_ft-1}, we obtain
\begin{align}
  &Z(\ft-1, x) = [ {\rm det} (I + Q_\ft B \Sigma_{\ft-1} B^\top )]^{-1/2} \exp \left( -\frac{1}{2} x^\top P_{\ft -1} x  \right), \\
  &P_{\ft -1} := Q_{\ft-1}  + A^\top P_\ft A  - A^\top P_\ft B (\Sigma_{\ft-1}^{-1} + B^\top P_\ft B)^{-1} B^\top P_\ft A , \ P_\ft := Q_\ft .
\end{align}
By applying the same argument as above for $ k = \ft-2,\ldots,0 $, we obtain the following result.
\begin{theorem}
  Assume that $ m = n $ and $ B $ is invertible.
  Then, the unique optimal policy $ \pi_k^* $ for Problem~\ref{prob:KLcontrol} with $ f(x,u) = Ax + Bu, \ \bbU = \bbR^n, \ w_k \sim \calN(0,\Sigma_k), \ \Sigma_k \succ 0 $, and \eqref{eq:quadratic_cost} is given by
  \begin{align}
    \pi_{k}^* ( u | x) &= \calN \bigl(u \bigl| -(\Sigma_{k}^{-1} + B^\top P_{k+1} B)^{-1} B^\top P_{k+1} Ax,  (\Sigma_{k}^{-1} + B^\top P_{k+1} B)^{-1} \bigr), \ k \in \bbra{\ft-1} \label{eq:lqg_policy}
  \end{align}
  where $ P_k $ is a solution of the Riccati difference equation
  \begin{align}
    &P_{k} = Q_{k}  + A^\top P_{k+1} A  - A^\top P_{k+1} B (\Sigma_{k}^{-1} + B^\top P_{k+1} B)^{-1} B^\top P_{k+1} A , \ k\in \bbra{\ft-1}, \label{eq:riccati} \\
    &P_\ft = Q_\ft \label{eq:riccati_terminal} .
  \end{align}
  The desirability function is given by
  \begin{align}\label{eq:Z_forward}
    Z(k,x) &= \left(\prod_{s=k+1}^\ft [{\rm det} (I + P_s B\Sigma_{s-1} B^\top)]^{-1/2}  \right)  \exp \left( -\frac{1}{2} x^\top P_k x  \right) .
  \end{align}
  \hfill $ \diamondsuit $
\end{theorem}

The mean of the optimal policy \eqref{eq:lqg_policy} coincides with the LQR controller~\cite{Lewis2012}. In other words, the optimal policy is the LQR feedback controller perturbed by additive Gaussian noise with zero mean and covariance matrix $ (\Sigma_k^{-1} + B^\top P_{k+1} B)^{-1} $.

In the above, we have analyzed the desirability function based on the backward equation \eqref{eq:Z_eq}. Hence, the obtained representation \eqref{eq:Z_forward} contains the solution of the backward Riccati difference equation.
For comparison, we calculate the desirability function based on the forward representation \eqref{eq:path_integral}.
Let $ \bar{x}_{k+1:\ft} := [\bar{x}_{k+1}^\top \ \cdots \ \bar{x}_\ft^\top]^\top $ and
\begin{align}
  &\bar{A}_k := \left[ A^\top \ (A^2)^\top \ \cdots \ (A^{k})^\top \right]^\top , \\
  &\Sigma_{k+1:\ft} := \diag (\Sigma_{k+1},\ldots,\Sigma_{\ft}), \\
  &Q_{k+1:\ft} := \diag (Q_{k+1},\ldots,Q_{\ft}) ,
  \\
  &L_k := \left[
  \begin{array}{ccccc}
  B & 0 & \cdots & \cdots & 0 \\
  AB & B & \ddots &  & \vdots \\
  A^2B & AB & B & \ddots & \vdots \\
  \vdots & \vdots & \ddots & \ddots & 0 \\
  A^{k-1} B & A^{k-2}B & \cdots & AB & B
  \end{array}
  \right] .
  \end{align}
Then, the conditional distribution of $ \bar{x}_{k+1:\ft} $ given $ \bar{x}_k = x $ is $ \calN (\bar{A}_{\ft-k} x, L_{\ft-k} \Sigma_{k+1:\ft} L_{\ft-k}^\top  )$.
By Corollary \ref{cor:path_integral},
\begin{align}
  &Z(k,x) = \exp \left( -\frac{1}{2} \|x\|_{Q_k}^2  \right)  \bbE \left[  \exp \left( - \frac{1}{2} \| \bar{x}_{k+1:\ft} \|_{Q_{k+1:\ft}}^2  \right)  \biggl| \bar{x}_k = x  \right]  \nonumber \\
  &= [{\rm det} (I + Q_{k+1:\ft} L_{\ft-k} \Sigma_{k+1:\ft} L_{\ft-k}^\top) ]^{-1/2} \nonumber \\
  &\quad \times\exp \biggl( -\frac{1}{2} x^\top \bigl( Q_k + \bar{A}_{\ft-k}^\top (Q_{k+1:\ft}^{-1} + L_{\ft-k} \Sigma_{k+1:\ft} L_{\ft-k}^\top)^{-1} \bar{A}_{\ft-k}  \bigr) x  \biggr) , \label{eq:Z_for}
\end{align}
where $ \|x\|_{Q} := (x^\top Q x)^{1/2} $ for $ Q \succ 0 $.
The fact that the desirability function can be expressed in two different ways \eqref{eq:Z_forward},\eqref{eq:Z_for} is similar to the fact that the value function for LQR control
\begin{align}
  V_{\rm LQR}(k,x) := &\inf_{\{u_s\}}  \ \frac{1}{2} \|x_\ft\|_{Q_\ft}^2 + \sum_{s=k}^{\ft-1} \frac{1}{2} \left(\|x_s\|_{Q_s}^2 + \|u_s\|_{\Sigma_s^{-1}}^2 \right) \nonumber \\
  &{\rm subj.~to} \ x_{s+1} = Ax_s + Bu_s, \ s\in [k,\ft-1] ,\ x_k = x \nonumber
\end{align}
can be written in the following two ways:
\begin{align}
  V_{\rm LQR}(k,x)=
  \begin{cases}
    \frac{1}{2} x^\top P_k x, \\
    \frac{1}{2} x^\top \bigl(Q_k + \bar{A}_{\ft-k}^\top (Q_{k+1:\ft}^{-1} + L_{\ft-k} \Sigma_{k+1:\ft} L_{\ft-k}^\top)^{-1} \bar{A}_{\ft-k} \bigr) x .
  \end{cases} 
\end{align}

\section{Numerical examples}\label{sec:example}
In this section, we illustrate the reformulated KL control through two examples.
\subsection{Linear quadratic case}
Consider the linear quadratic case where
\begin{equation}
  A = 0.85, \ B  = 0.10, \ Q_k = 3.0, \ \Sigma_k = 1.5, \forall k
\end{equation}
and a finite horizon $ \ft  = 30 $.
First, for comparison we compute the associated value function in two ways: by using the explicit expression \eqref{eq:Z_forward} and by using a Monte Carlo method based on the path integral representation~\eqref{eq:path_integral}.
For the Monte Carlo method, we generate $ S $ sample paths $ \{\bar{x}_k^{(i)}\}_{k=0}^\ft, i=1,\ldots,S $ with $ \bar{x}_0 = x, w_k \sim \calN(0, \Sigma_k) $ and compute
\[
  -\log \left[\frac{1}{S} \sum_{i=1}^S \exp\left(-\sum_{s=0}^\ft \ell_s(\bar{x}_s^{(i)})\right) \right]
\]
to approximate $ V(0,x) $.
As shown in Fig.~\ref{fig:value_func}, $ V(0,x) $ is well approximated by the Monte Carlo estimate with $ 3000 $ samples. The computation time for each $ x $ is about $ 0.025 $ s, $ 0.24 $ s, and $ 0.71 $ s for $ S = 100,1000,3000 $, respectively, with MATLAB on MacBook~Pro with Apple~M1~Pro.
Note that the Monte Carlo simulations can be easily parallelized.
Next, three samples of the optimal state and control processes $ \{x_k\}, \{u_k\} $ for different $ (Q_k,\Sigma_k) $ are shown in Fig.~\ref{fig:trajectory}.
As can be seen, as $ \Sigma_k $ increases, the absolute value of mean and variance of the optimal control gets larger.
This is because for larger $ \Sigma_k $, the cost of shifting the transition distribution $ \pdf_{k+1}^{\pi_k}(\cdot | x_{k}) $ from the reference distribution $ \calN(\cdot | Ax_k, B\Sigma_k B^\top) $ becomes smaller, while the cost of reducing the variance of the transition distribution becomes larger.
In Fig.~\ref{fig:lq_q3sig10} and Fig.~\ref{fig:lq_q30sig1}, the values of $ Q_k/\Sigma_k^{-1} $ coincide. Therefore the mean values of the optimal policies~\eqref{eq:lqg_policy} for the two cases also coincide although the control process in Fig.~\ref{fig:lq_q30sig1} has smaller variance than in Fig.~\ref{fig:lq_q3sig10}.
On the other hand, for the LQR problem whose cost is given by
\[
  \frac{1}{2} Q x_\ft^2 + \sum_{k=0}^{\ft-1} \frac{1}{2}(Q x_k^2 + \Sigma^{-1} u_k^2 ),
\]
the optimal control depends on $ Q, \Sigma $ only via $ Q/\Sigma^{-1} $.
This is in clear contrast to KL control.

\begin{figure}[tb]
	\centering
	\includegraphics[scale=0.4]{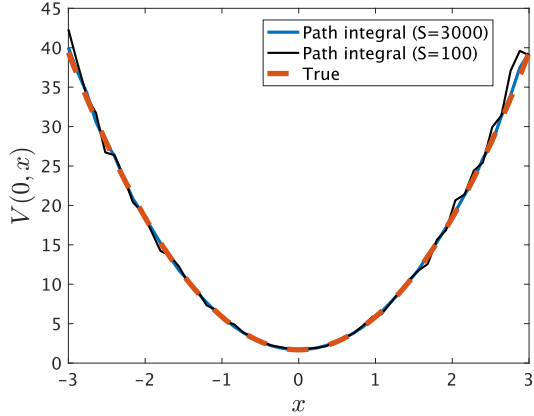}
	\caption{Monte Carlo estimates of the value function $ V(0,x) $ (red, dashed) with $ S  =100 $ (black) and $ S = 3000 $ (blue).}
	\label{fig:value_func}
\end{figure}

\begin{figure}[t]
  \centering
  \subfloat[$ Q_k = 3.0, \ \Sigma_k = 0.5 $]{\includegraphics[keepaspectratio, scale=0.23]
  {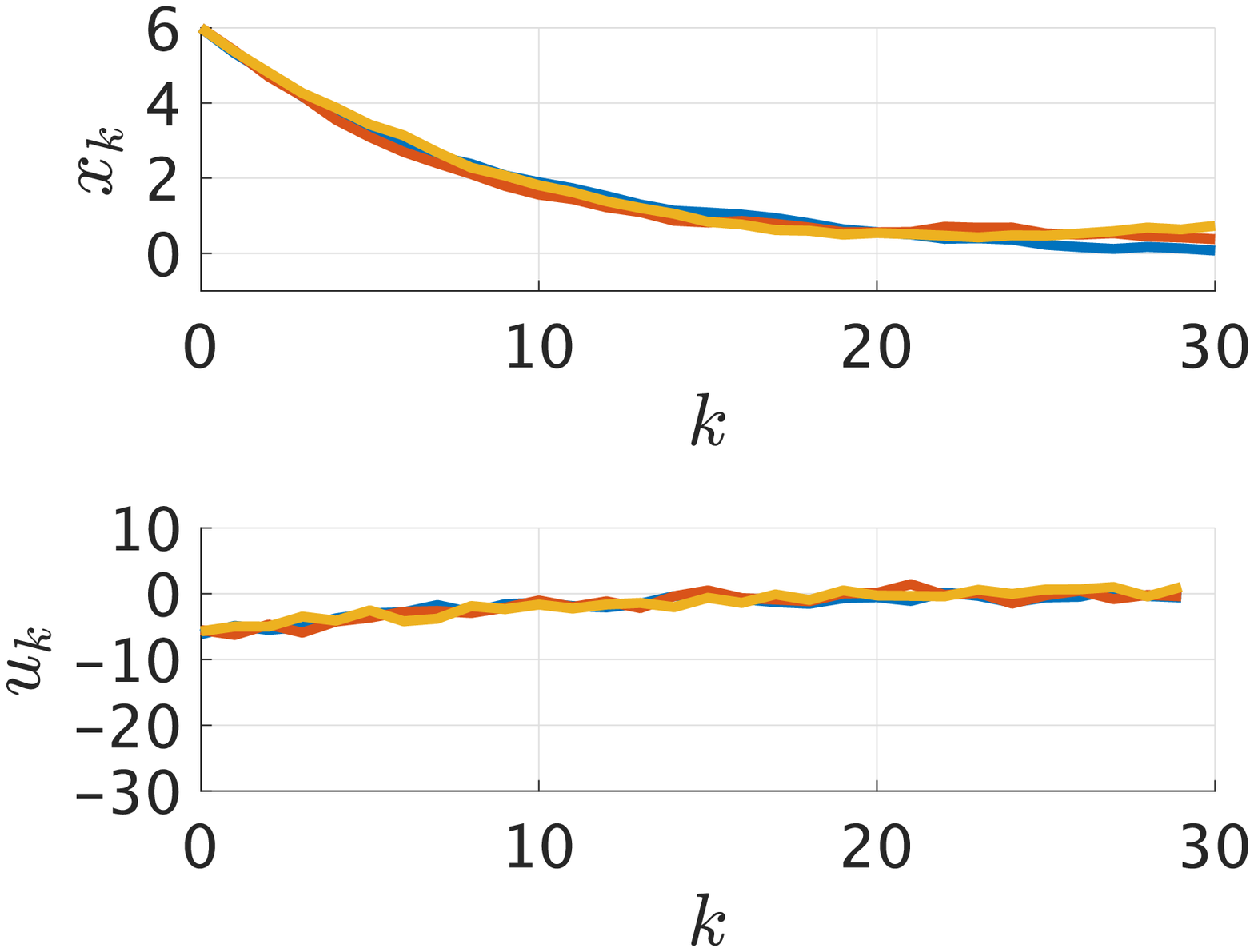}} \quad
  \subfloat[$ Q_k = 3.0, \ \Sigma_k = 1.5 $]{\includegraphics[keepaspectratio, scale=0.23]
  {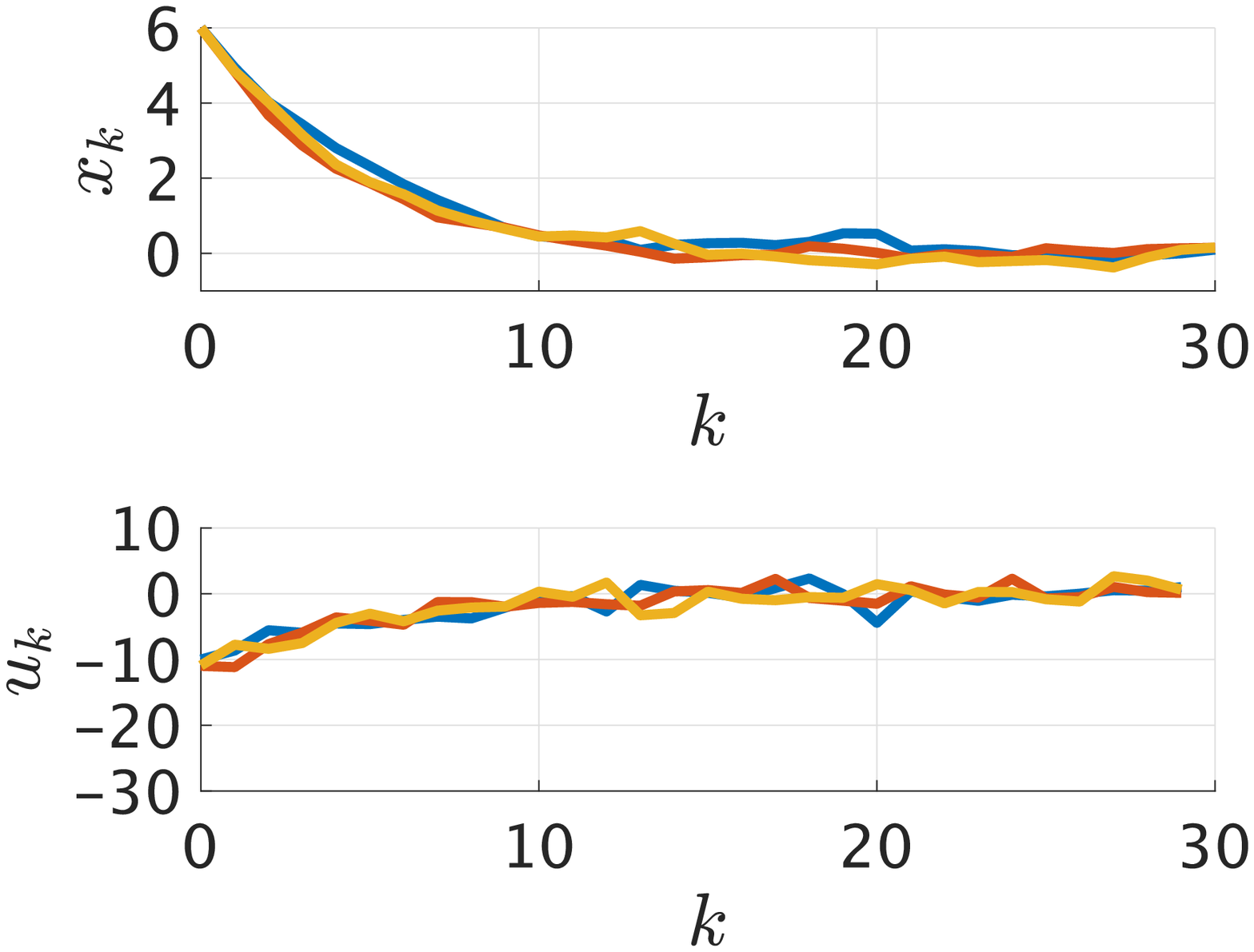}}
  \\
  \centering
  \subfloat[$ Q_k = 3.0, \ \Sigma_k = 10.0 $]{\includegraphics[keepaspectratio, scale=0.23]
  {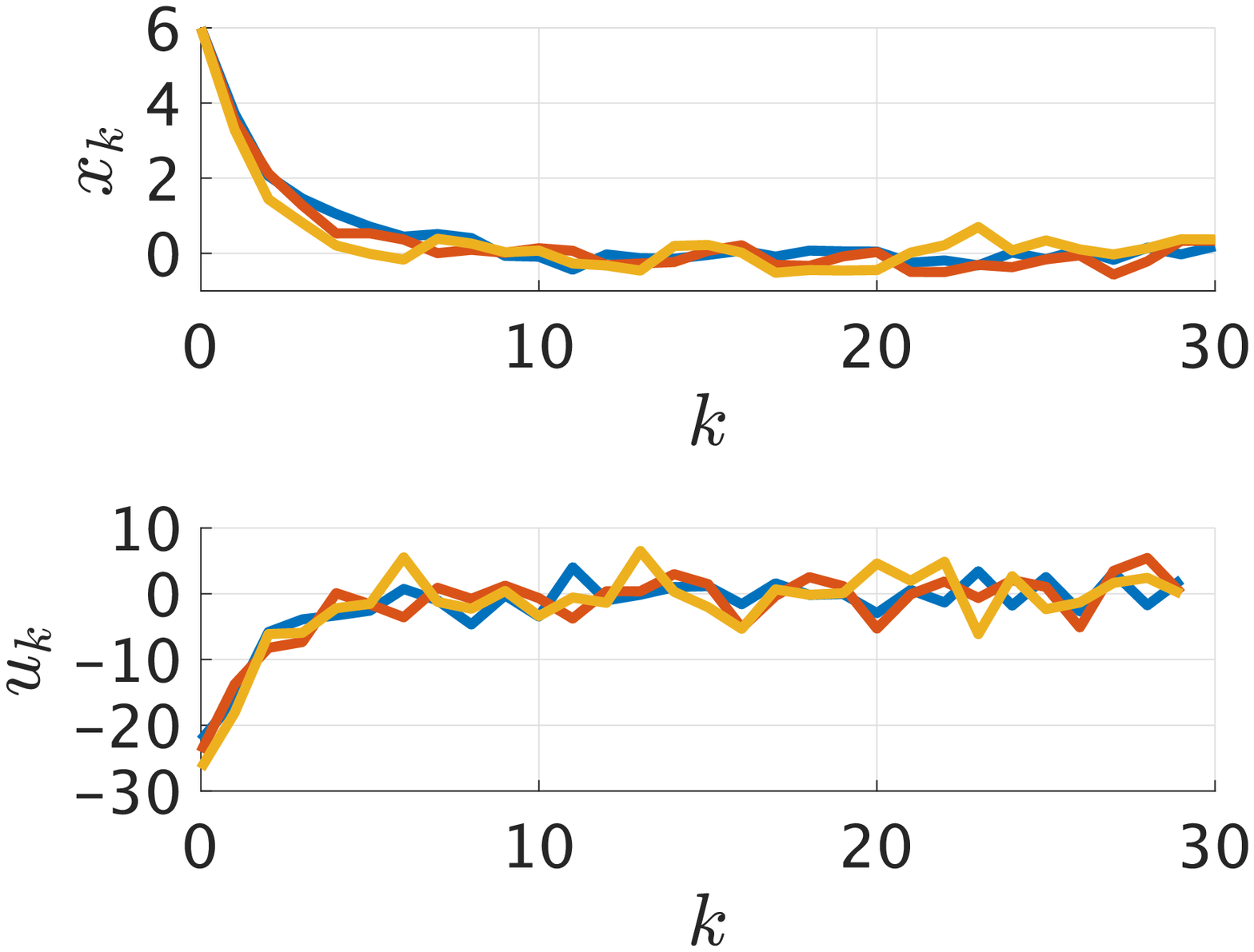}\label{fig:lq_q3sig10}} \quad
  \subfloat[$ Q_k = 30, \ \Sigma_k = 1.0 $]{\includegraphics[keepaspectratio, scale=0.23]
  {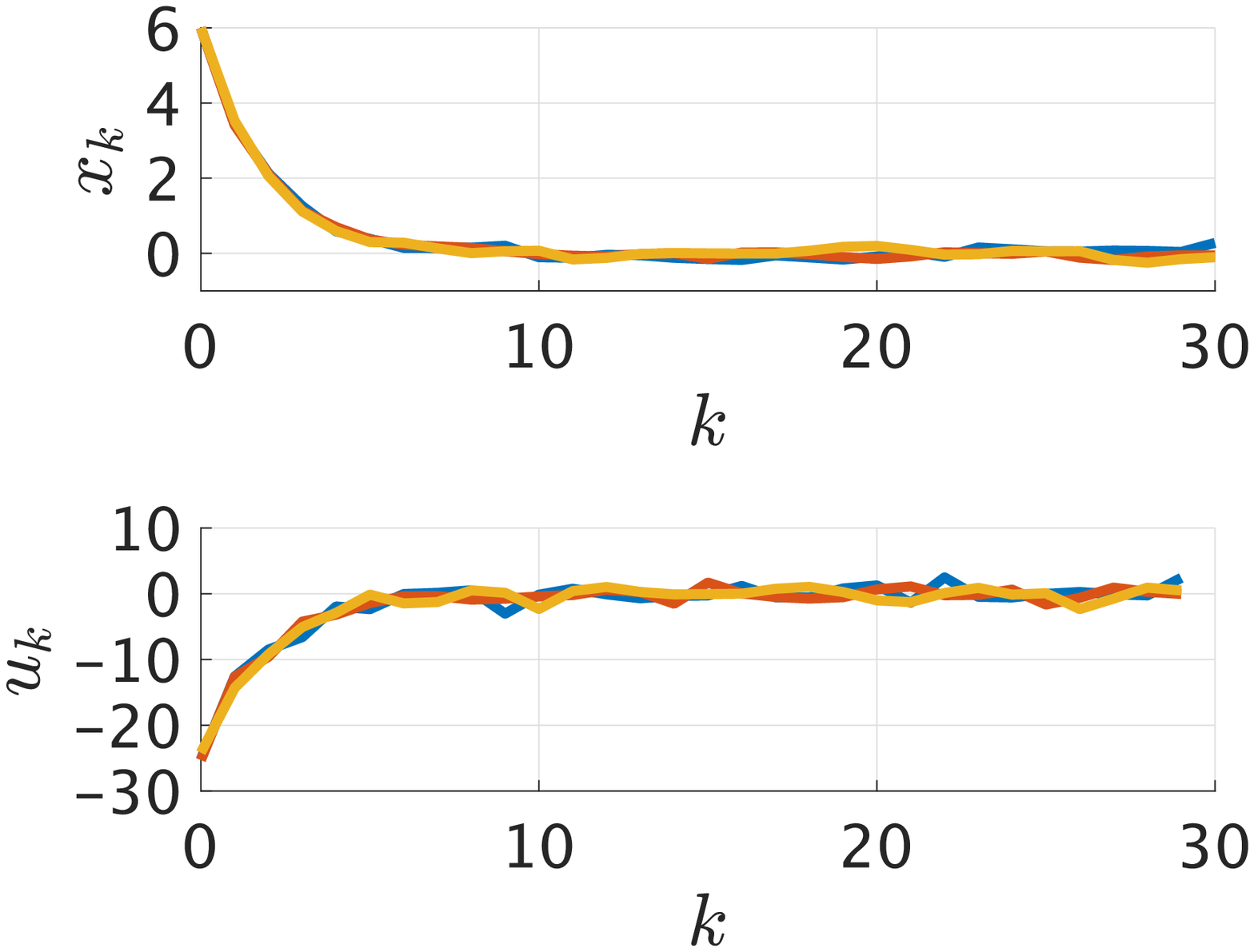}\label{fig:lq_q30sig1}}
  \caption{Three samples of the optimal state and control processes $ \{x_k\},\{u_k\} $ for different $ (Q_k, \Sigma_k) $.}\label{fig:trajectory}
\end{figure}

\subsection{Cart-pole pendulum}
We now proceed to apply our result to a nonlinear system. Specifically, we consider the cart-pole inverted pendulum in Fig.~\ref{fig:cart_pole}.
The system consists of a cart of mass $ M = 1.0~{\rm kg}$ moving horizontally, a massless rod of length $ L = 1.0~{\rm m} $ attached to the cart and rotating around a pivot point in the $ \bar{x}y $-plane only, and a point mass $ m=0.1~{\rm kg} $ at the end of the rod.
The input $ u $ is the horizontal force applied to the cart to maintain the pendulum in a balanced and upright position.
Here, we neglect the influence of friction.
Let $ \bar{x},  \theta $ be the position of the cart and the angle of the rod ($ \theta = 0 $ for the upright position and $ \theta = \pi $ for the downward position of the pendulum), respectively.

We then have the following continuous-time model of the cart-pole system:
\begin{align}
  &\ddot{\bar{x}} = \frac{-mL (\dot{\theta})^2 \sin \theta + mg \sin\theta \cos \theta + u}{M + m\sin^2 \theta} =: h_1 (\theta, \dot{\theta}, u) , \\
  &\ddot{\theta} = \frac{1}{L} ( h_1(\theta, \dot{\theta}, u) \cos \theta + g \sin \theta ) =: h_2 (\theta,\dot{\theta}, u),
\end{align}
where $ g = 9.8 ~{\rm m/s^2}$ is the gravitational acceleration.
By the Euler method, we obtain the discrete-time system:
\begin{align}
  x_{k+1} = f(x_k,u_k) =
  \begin{bmatrix}
    \bar{x}_k + \tau \dot{\bar{x}}_k \\
    \dot{\bar{x}}_k + \tau h_1 (\theta_k, \dot{\theta}_k, u_k) \\
    \theta_k + \tau \dot{\theta}_k \\
    \dot{\theta}_k + \tau h_2 (\theta_k, \dot{\theta}_k, u_k ) 
  \end{bmatrix} ,
\end{align}
where $ x_k = [\bar{x}_k \ \dot{\bar{x}}_k \ \theta_k \ \dot{\theta}_k]^\top $ and $ \tau = 0.05~{\rm s} $.
Here, we consider the discrete input space $ \bbU = \{ 2i~{\rm N} \}_{i=-10}^{10}$. For a cost function, let
\[
  \ell_k (x_k) =  q_1 |\bar{x}_k| + q_2 |\dot{\bar{x}}_k| + q_3 |\theta_k| + q_4 |\dot{\theta}_k|
\]
with $ q_1 = 7.0~{\rm m^{-1}}, \ q_2 = 2.5~{\rm s/m}, \ q_3 = 7.0~{\rm rad}^{-1}, \ q_4 = 2.5~{\rm s/rad} $. In addition, the noise $ w_k $ is designed to follow a discretized Gaussian distribution
\begin{equation}
  \bbP (w_k = {\rm w}) \propto \exp \left( - \frac{1}{2\sigma^2} {\rm w}^2  \right), \ {\rm w} \in \bbW = \bbU 
\end{equation}
with $ \sigma = 5.0~{\rm N} $. The initial state is given by $ \bar{x}_0 = 2.0~{\rm m}, \ \dot{\bar{x}}_0 = 0~{\rm m/s}, \ \theta_0 = 0.5~{\rm rad}, \ \dot{\theta}_0 = 0~{\rm rad/s} $.

Suppose that the state value at the current time $ k $ is $ x_k = x $.
Then by Corollary~\ref{cor:discrete}, the optimal policy at time $ k $ is given by
\begin{equation}
  \Pi_k^* (u | x) \propto \bbP\left(w_k = u \right) Z\left(k+1, f(x,u) \right), \ u \in \bbU,
\end{equation}
where the desirability function $ Z(k+1, f(x,u)) $ for each $ u \in \bbU $ can be computed by the Monte Carlo method based on \eqref{eq:path_integral}.
In this example, we use $ 5000 $ samples for the sampling approximation of $ Z $.

Figure~\ref{fig:position_angle} shows $ 50 $ sample paths of the optimal state process in the $ \bar{x} \theta $-plane.
The optimal policy balances the pendulum around the upright position while the cart-pole system fluctuates around the origin due to the stochasticity of the policy.
The detailed behavior of the optimal state and control processes is illustrated in Fig.~\ref{fig:time_series}.
One can see that the cart and pole velocity shows large fluctuations while as $ k $ increases, their mean values approach zero.
If one takes larger values of $ q_2,q_4 $, their fluctuations are reduced.

\begin{figure}[tb]
	\centering
	\includegraphics[scale=0.1]{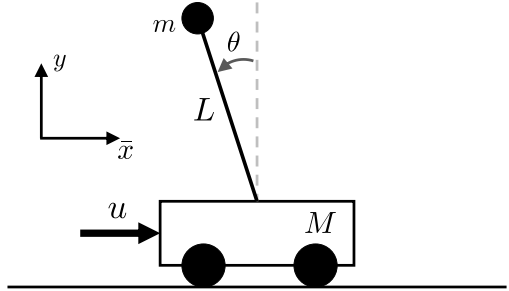}
	\caption{Cart-pole pendulum.}
	\label{fig:cart_pole}
\end{figure}

\begin{figure}[tb]
	\centering
	\includegraphics[scale=0.35]{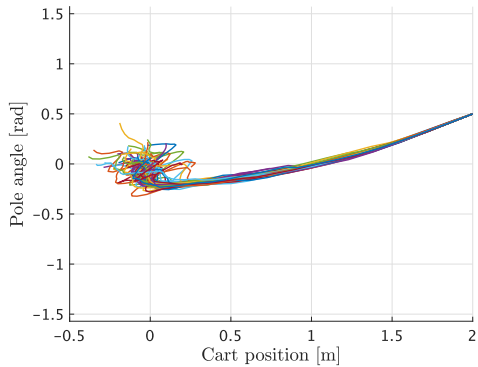}
	\caption{$ 50 $ sample paths of the optimal state process $ \{(\bar{x}_k, \theta_k)\} $.}
	\label{fig:position_angle}
\end{figure}

\begin{figure}[tb]
	\centering
	\includegraphics[scale=0.33]{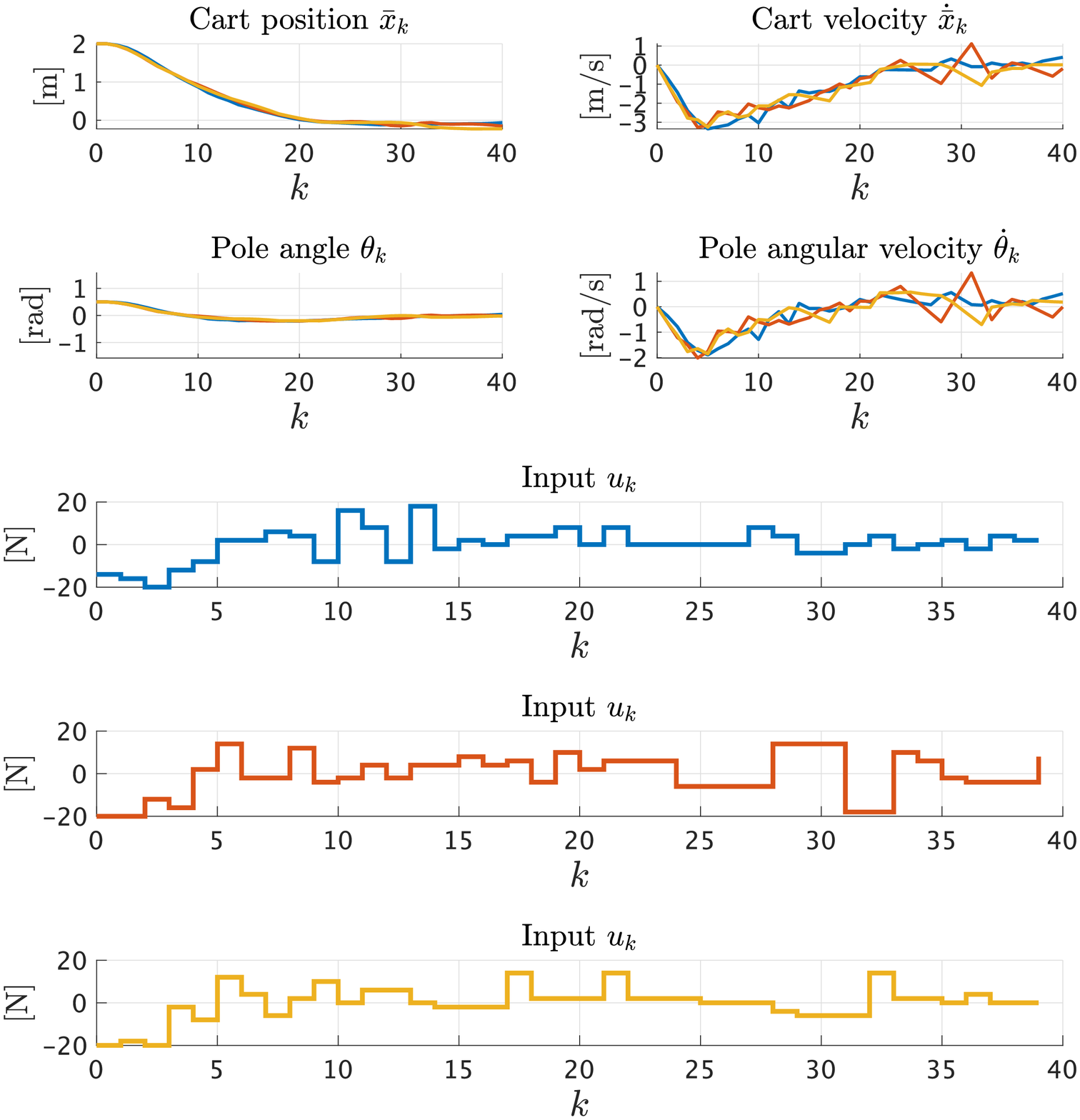}
	\caption{Three sample paths of the optimal state and control processes for the cart-pole pendulum. The same color indicates the correspondence between the sample paths of the state process and the control process.}
	\label{fig:time_series}
\end{figure}

\section{Conclusions}\label{sec:conclusion}
In this paper, we reformulated KL control to make its assumption reasonable for continuous spaces and remove the approximation of dynamics.
Then, we analyzed the associated optimal control via the desirability function.
In particular, we showed that the reformulated KL control admits \magenta{sampling approximations of the desirability function.
We emphasize that the Bellman equation for the infinite horizon KL control can also be linearized by the same argument as in the finite horizon case, and the associated inverse reinforcement learning can be formulated as a convex optimization~\cite{Dvijotham2010}.
}
In addition, we revisited the original KL control and clarified that the assumption of controllability of transition densities implies the noncausality of controllers.
For linear systems with a quadratic state cost and Gaussian noise, we derived the optimal policy analytically.
Lastly, we illustrated our KL control via numerical examples.
Future work will focus on weakening Assumptions~\ref{ass:invertible}-(ii),(iii) by analyzing the problem without using densities.

\section*{Acknowledgements}
This work was supported in part by JSPS KAKENHI Grant Numbers JP21J14577, JP21H04875, and by JST, ACT-X Grant Number JPMJAX2102.

\section*{Disclosure statement}
No potential conflict of interest was reported by the authors.

\section*{Funding}
This work was supported in part by JSPS KAKENHI [grant
number JP21J14577, JP21H04875] and by JST, ACT-X [grant number
JPMJAX2102].









\bibliographystyle{tfnlm}
\bibliography{JCMSI_KL}




\end{document}